\documentclass[%
 amsmath,amssymb,
 aps,
10pt ]{revtex4-2}

\usepackage{graphicx}
\usepackage{dcolumn}
\usepackage{bm}
\usepackage{ulem}
\usepackage{xcolor}
\usepackage{amsmath}
\usepackage{amsthm}
\usepackage{amssymb,amsfonts}

\usepackage{hyperref}

\def\sideremark#1{\ifvmode\leavevmode\fi\vadjust{\vbox to0pt{\vss
 \hbox to 0pt{\hskip\hsize\hskip1em
 \vbox{\hsize2cm\tiny\raggedright\pretolerance10000
  \noindent #1\hfill}\hss}\vbox to8pt{\vfil}\vss}}}

\newcommand{\ext}{{\rm{d}}}

\newtheorem{definition}{Definition}
\newtheorem{proposition}{Proposition}
\newtheorem{corollary}{Corollary}
\theoremstyle{remark}
\newtheorem*{remark}{Remark}

\newcommand{\mbold}[1]{\mbox{\boldmath{\ensuremath{#1}}}}

\def \bell {\mbox{{\mbold\ell}}}

\def \bk {\mbox{{\bf k}}}
\def \bm {\mbox{{\bf m}}}

\def \bcm {\mbox{{\bf M}}}

\def \bcl {\mbox{{\bf L}}}
\def \bck {\mbox{{\bf K}}}
\def \bgam {\mbox{{\mbold \Gamma}}}

\def\beq#1\eeq{\begin{align}#1\end{align}}

\begin{document}

\preprint{APS/123-QED}


\title{Unique Carrollian manifolds emerging from Einstein spacetimes}

\author{Samuel Blitz}
\affiliation{
 Department of Mathematics and Statistics, Masaryk University\\ Building 08, Kotlářská 2,
Brno, CZ 61137 \\
blitz@math.muni.cz
 }%

\author{David McNutt}\thanks{Corresponding author}
\affiliation{ Center for Theoretical Physics, Polish Academy of Sciences, Warsaw \\ 
mcnuttdd@gmail.com
}%

\author{Pawel Nurowski}
\affiliation{ Center for Theoretical Physics, Polish Academy of Sciences, Warsaw \\
nurowski@fuw.edu.pl
}%

\date{\today}

\begin{abstract}

We explicitly determine all shear-free null hypersurfaces embedded in an Einstein spacetime, including vacuum asymptotically flat spacetimes. We characterize these hypersurfaces as oriented 3-dimensional manifolds where each is equipped with a coframe basis, a structure group and a connection. Such manifolds are known as null hypersurface structures (NHSs). The coframe and connection one-forms for an NHS appear as solutions to the projection of the Cartan structure equations onto the null hypersurface. We then show that each NHS corresponds to a Carrollian structure equipped with a unique pair of Ehresmann connection and affine connection.

\end{abstract}

\maketitle


\section{\label{sec:intro}Introduction\protect}

\bigskip

Within general relativity (GR) and more generally in Lorentzian geometry, null hypersurfaces are defined as hypersurfaces whose normal vector-field are non-zero but with zero magnitude with respect to the metric. Unlike timelike or spacelike hypersurfaces, where the normal vector-field has positive or negative magnitude, the induced metric for null hypersurfaces is necessarily degenerate \cite{penrose1972}. Furthermore, the null normal vector-field lies in the tangent space of the hypersurface.     

Despite this degeneracy, null hypersurfaces play an important role in GR and other gravity theories.  For black hole solutions, the event horizon of the black hole is a null hypersurface \cite{Ashtekar2002, wald2010general}. In the study of asymptotically flat Einstein spacetimes, Newman and Unti showed that on a foliation of null hypersurfaces, the Einstein field equations reduce to a set of first-order linear inhomogeneous differential equations \cite{newman1962behavior}. Furthermore, asymptotic null infinity can, unsurprisingly, be treated as  a null hypersurface using conformal methods \cite{Herfray2022}.In the construction of new spacetimes coming from gluing solutions together using junction conditions, null hypersurfaces are the natural choice for black hole and wormhole solutions \cite{mars1993}. In the case of degenerate Killing horizons, and in particular extremal Killing horizons, these hypersurfaces allow for the construction of near-horizon geometries (NHGs). These spacetimes represent the open region in the proximity of the black hole's region and give important information about the physics and topology of these black holes \cite{kunduri2013classification}. Strictly speaking, NHGs belong to the degenerate Kundt class of spacetimes \cite{coley2009kundt} while the majority of black hole solutions do not \cite{kramer}. Therefore, the behavior and properties of the null hypersurface are not necessarily the same as those of the original horizon.

Due to the degeneracy of the induced metric on null hypersurfaces, they cannot be realized as Riemannian geometries~\cite{Henneaux1979,figueroa2020,pinzanifokeeva2022,blitz2023horizons}. Instead, the geometric structure is called a Carrollian geometry, whose structure group is the Carroll group \cite{LvyLeblond1965}. This group is a dual contraction of the Poincar{\'e} group where the velocity of light is set to zero. As these are physically relevant hypersurfaces in spacetime, there is a growing interest in applications of Carrollian geometries in black hole physics~\cite{penna2016,Donnay_2019,redondo2023,gray2023carrollian,ciambelli_fluid-gravity_2023,freidel2024,ciambelli2024} and field theories~\cite{bagchi_flat_2016,Ciambelli2018,Bagchi2019,Bagchi2020,gupta_constructing_2021,Banerjee2021,rivera2022revisiting,donnay_carrollian_2022,Bagchi2022,chen_higher-dimensional_2023,bagchi_carroll_2023,baiguera2023conformal,chen_constructing_2023,mehra_toward_2023,salzer_embedding_2023,banerjee_one-loop_2023}. In particular, it has been shown that fluids lying on black hole horizons have non-trivial dynamics, known as Carrollian hydrodynamics \cite{bergshoeff2014dynamics, freidel2023carrollian, freidel2024}. Furthermore, there is a growing interest in the behaviour of particles that lie on the horizon of a black hole. The dynamics of these particles are dictated by the Carroll group. In the case of particles on the Kerr-Newman black hole's horizon, two different approaches were used to study the behaviour of particles on the horizon, notably the anyonic spin-hall effect \cite{marsot2022anyonic,gray2023carrollian}. While the horizon of the Kerr-Newman black hole is outside of the scope of the current paper, this example motivates the necessity for a unique approach to determining a null hypersurfaces geometry as a Carrollian geometry.

Mathematically, a Carrollian geometry is a $(d-1)$-dimensional differentiable manifold $\mathcal{N}$ equipped with a rank-$(d-2)$ semi-definite symmetric bilinear form ${\bf h}$ and a \textit{fundamental vector field} ${\bk}$ that spans the kernel of ${\bf h}$. In this paper, we consider those Carrollian geometries induced by null hypersurfaces in four dimensional spacetimes, where $d=4$.

To discuss such structures concretely, more data is required. With the addition of an Ehresmann connection $\bck$ that is dual to $\bk$, the tangent space $T \mathcal{N}$ is decomposed into a horizontal subbundle $H \mathcal{N}$ and the span of $\bk$. Furthermore, when $\bck$ is integrable, the horizontal subspace at a point can be integrated into a spacelike submanifold $\mathcal{C}$ whose tangent space $T_p \mathcal{C}$ at a point $p \in \mathcal{C}$ is the horizontal subspace through that point $H_p \mathcal{N}$. Indeed, this picture explains why a Carrollian geometry paired with an Ehresmann connection is called \textit{ruled}~\cite{ciambelli2024}, as it defines a foliation of the Carrollian geometry. 

Within Carrollian geometry, there is substantial debate on the choice of affine connection used to construct differential geometric invariants. Following from arguments using Spencer cohomology, there is no unique torsion-free connection for Carrollian manifolds constructed only from the intrinsic geometry of the manifold\cite{figueroa2020}, unlike for (pseudo-)Riemannian manifolds. As such, different geometric investigations may warrant different choices of additional geometric constraints that determine an affine connection.

In the literature, two choices seem to appear: an affine connection that is torsion-free but non-metric~\cite{Chandrasekaran2022,mars2023,mars2023-2,ciambelli2024,freidel2024,freidel2024-2} and an affine connection that is metric but has torsion~\cite{Hartong2015,bergshoeff2017,bekaertmorand2018,Ciambelli2018,Ciambelli2018-2,Ciambelli2019,figueroa2020,gallegos2021,figueroa2022,miskovic2023,marotta2023,baiguera2023conformal,blitz2023horizons}. While the latter is perhaps more well-motivated mathematically, the former seems to be more useful in physical contexts.

It is not clear, however, that these two options completely exhaust the natural geometric construction of affine connections on null hypersurfaces. To that end, we construct an affine connection for the Carrollian geometry in a new way. In~\cite{nurowski2000}, Nurowski and Robinson used Cartan's moving frame approach to study the invariants of general null hypersurfaces, called null hypersurface structures (NHS), which included an examination of the bulk connection forms pulled back to the hypersurface. Inspired by their work, in this paper we examine the induced connection (found by gauge-fixing the freedom in the frame) on null hypersurfaces embedded in bulk spacetimes that are shear-free solutions to the vacuum Einstein equations. While somewhat restrictive, this family of hypersurfaces accounts for many solutions of physical interest (as, for example, many black hole spacetimes solve the vacuum Einstein equations). In doing so, we are able to produce closed form expressions for the connection coefficients and curvature components for all null hypersurfaces that can possibly be embedded in such spacetimes. The resulting induced connection can then be compared to the two connections described above. Furthermore, the intrinsic and extrinsic first-order invariants of a Carrollian geometry can be fully characterized for all null hypersurfaces satisfying the constraints we have imposed.

\section{Cartan Structure equations and their pullback to a null hypersurface} \label{sec2}

We will consider the Newman-Penrose (NP) formalism as presented in \cite{nurowski2000} for a four-dimensional spacetime $(M, {\bf g})$ with a metric ${\bf g}$ and signature $(+++-)$. We define $({\bf e}_1, {\bf e}_2,{\bf e}_3,{\bf e}_4) = (\bm, \bar{\bm}, \bell, \bk)$ with the dual coframe $(\theta^1, \theta^2, \theta^3, \theta^4) = (\bcm, \bar{\bcm}, \bcl, \bck)$ so that 
\begin{eqnarray*}
	&\bcm = {\bf g}(\bar{\bm},\cdot),\quad \bar{\bcm} = {\bf g}({\bm},\cdot), \quad \bcl = {\bf g}(\bk,\cdot), \quad \bck = {\bf g}(\bell,\cdot).
\end{eqnarray*}
\noindent Relative to this basis of one-forms the metric takes the form
\beq
    {\bf g} = -2\bck \bcl + 2\bcm \bar{\bcm}, \label{eqn:gcoframe}
\eeq

\noindent With this choice the Lorentz frame transformations that preserve the null direction $\bk$ can be written as 

\begin{equation}
\begin{aligned}
    \bcm'  &= e^{i\phi}[\bcm + z \bcl ], \\
    \bcl' &= A^{-1}\bcl, \\
    \bck' &= A[\bck + z \bar{z} \bcl + z \bar{\bcm} + \bar{z}\bcm]. 
\end{aligned} \label{grp:kfixed}
\end{equation}
\noindent where $\phi$ and $A$ are real-valued parameters while $z$ and its conjugate are complex-valued parameters. 

The first Cartan structure equations are an exterior differential system that relates the exterior derivatives of the coframe basis to the components of the connection one-forms $\bgam_{ab}$:
\begin{equation}
\begin{aligned}
	& \ext \bcm = - \bgam_{21} \wedge \bcm - \bgam_{23} \wedge \bcl - \bgam_{24} \wedge \bck, \\
	& \ext\bcl = \bgam_{41} \wedge \bcm + \bgam_{42} \wedge \bar{\bcm} + \bgam_{43} \wedge \bcl,\\
	&\ext \bck = \bgam_{31} \wedge \bcm + \bgam_{32} \wedge \bar{\bcm} + \bgam_{34} \wedge \bck. 
\end{aligned} \label{CartanStr1}
\end{equation}

\noindent where 
\begin{equation}
\begin{aligned}
	& {\bf \bgam}_{41} = \sigma \bcm + \rho \bar{\bcm} + \tau \bcl + \kappa \bck, \\
	& \bgam_{23} = \mu \bcm + \lambda \bar{\bcm} + \nu \bcl + \pi \bck, \\
	& \frac12 (\bgam_{12} + \bgam_{34}) = -\beta \bcm - \alpha \bar{\bcm} - \gamma \bcl - \epsilon \bck. 
\end{aligned} \label{Bulk1forms}
\end{equation}
\noindent The coefficients of the one-forms are complex-valued.  It is important to note here that traditionally the connection one-forms are written as $\bgam^a_{~b} = \Gamma^a_{~bc} \theta^c$. However, the index $a$ may be lowered to exploit the metric compatibility condition $\bgam_{ab} = - \bgam_{ba}$. For example, $\bgam^1_{~3}$ is related to $\bgam^4{}_2$ since both can be expressed as $\bgam_{23}$.

Under the transformations in equation \eqref{grp:kfixed}, these tranform as:
\begin{equation}
    \begin{aligned}
        \bgam_{41}' & = A^{-1}e^{-i\phi}\bgam_{41},\\
        \bgam_{23}' &=A e^{i\phi}(\bgam_{23} + z [\bgam_{12}+\bgam_{34}] - z^2\bgam_{14}-\ext z),\\
        [\bgam_{12}+\bgam_{34}]' &= \bgam_{12}+\bgam_{34}+2z \bgam_{41}+\ext (\ln A) + i \ext \phi.
    \end{aligned} \label{eqn:GammaTransfs}
\end{equation}

Denoting the Weyl spinor scalars as
\begin{equation}
\begin{aligned}
    & \Psi_0 = C_{abcd}\bcl^a\bar{\bcm}^b \bcl^c \bar{\bcm}^d, \quad \Psi_3 = C_{abcd}\bcl^a\bck^b {\bcm}^c \bck^d, \\
    & \Psi_1 = C_{abcd}\bcl^a\bck^b \bcl^c \bar{\bcm}^d, \quad \Psi_4 = C_{abcd} {\bcm}^a\bck^b {\bcm}^c \bck^d, \\
    & \Psi_2 = C_{abcd}\bcl^a\bar{\bcm}^b {\bcm}^c \bck^d = \frac12 C_{abcd}\bcl^a\bck^b(\bcl^c\bck^d -\bar{\bcm}^c{\bcm}^d),
\end{aligned}    
\end{equation}
\noindent the Ricci scalar as $R=R^i_{~i}$, and the traceless Ricci tensor as $S_{ij} = R_{ij} - g_{ij}R/4~~i,j = 1,\ldots, 4$, the second Cartan structure equations are~\cite{kramer}:
\begin{widetext}
\begin{equation}
\begin{aligned}
	\ext\bgam_{23} &= (\bgam_{12} + \bgam_{34}) \wedge \bgam_{23} + \Psi_4 \bar{\bcm} \wedge \bcl + \Psi_3 (\bck \wedge \bcl - \bcm \wedge \bar{\bcm}) \\
	& \quad  + (\Psi_2 + \frac{1}{12} R) \bck \wedge \bcm + \frac{1}{2} S_{33} \bcm \wedge \bcl + \frac{1}{2} S_{32} (\bck \wedge \bcl + \bcm \wedge \bar{\bcm})  + \frac{1}{2}S_{22} \bck \wedge \bar{\bcm}, \\
	\ext \bgam_{14} &= \bgam_{14} \wedge (\bgam_{12} + \bgam_{34}) + (-\Psi_2 - \frac{1}{12}R) \bar{\bcm} \wedge \bcl - \Psi_1 (\bck \wedge \bcl - \bcm \wedge \bar{\bcm}) \\
	& \quad   - \Psi_0 \bck \wedge \bcm - \frac{1}{2} S_{11} \bcm \wedge \bcl - \frac{1}{2} S_{41}( \bck \wedge \bcl + \bcm \wedge \bar{\bcm}) - \frac{1}{2}S_{44} \bck	\wedge \bar{\bcm}, \\
		 \frac{1}{2}(\ext \bgam_{12} +\ext\bgam_{34}) &= \bgam_{23} \wedge \bgam_{14} - \Psi_3 \bar{\bcm} \wedge \bcl - (\Psi_2 - \frac{1}{24}R) (\bck \wedge \bcl - \bcm \wedge \bar{\bcm}) \\ 
		& \quad  - \Psi_1 \bck \wedge \bcm - \frac{1}{2} S_{31} \bcm \wedge \bcl - \frac{1}{4}(S_{12} + S_{34})(\bck \wedge \bcl + \bcm \wedge \bar{\bcm}) -\frac{1}{2}S_{42} \bck \wedge \bar{\bcm}. 
\end{aligned} \label{CartanStr2}
\end{equation}
\end{widetext}

Without loss of generality, a degenerate 3-manifold $C$ is realized as a null hypersurface embedded in the spacetime M as $\mathcal N=\varphi(C)$, via the embedding
$\varphi: C\to M$ such that ${\varphi}^*{\bf \bcl} = 0$. Similarly, using the natural projection operator to $T \mathcal{N}$  induced by the choice of frame, $\top : TM|_{\mathcal{N}} \rightarrow T \mathcal{N}$, we have that $\top \bell = 0$. 
Thus using index notation, for any tensorial quantity, components with an index equal to 3 must be zero when viewed as objects on the hypersurface.
This observation is key to deriving the intrinsic classification of null hypersurfaces as Carrollian manifolds. Then the Carrolian structure ${\bf h},{\bf k}$ on $C$ is given by ${\bf h}=\varphi^* {\bf g}=2{\bf M} \bar{\bcm}$ and ${\bf k}$ being the degenerate direction for ${\bf h}$, and the structure equations on $C$ are the pullbacks via $\phi$ of Cartan's first and second structure equations for the spacetime.

For example, by combining equations \eqref{CartanStr1} with \eqref{Bulk1forms}, we find that 
\beq 
   \ext \bcl = (\bar{\rho}-\rho) \bcm \wedge \bar{\bcm}+ (\bar{\alpha}+\beta-\tau) \bcm \wedge \bcl + (\alpha + \bar{\beta}-\bar{\tau}) \bar{\bcm} \wedge \bcl + \kappa \bck \wedge \bcm + \bar{\kappa} \bck \wedge \bar{\bcm} + (\epsilon + \bar{\epsilon} )\bck \wedge \bcl. 
\eeq

\noindent Pulling back this equation to the null hypersurface, the left-hand side must vanish, whereas on the right-hand side only select terms vanish. Thus we find that
\beq 
    Im(\rho)|_{\mathcal{N}} = 0,\quad \kappa|_{\mathcal{N}} =0, 
\eeq
\noindent implying that ${\bf k}$ is hypersurface-orthogonal and geodesic. As an abuse of notation in what follows, when obvious from context, we will omit $|_{\mathcal{N}}$ and work with the same symbols when working with pulled back quantities.

Instead of describing an NHS as a three-dimensional oriented manifold $\mathcal{N}$ equipped with a degenerate metric, ${\bf h}$ of signature $(++0)$, we may instead describe it as an oriented manifold $\mathcal{N}$ equipped with a basis of coframe fields. Recalling that in the spacetime, we may adapt our null coframe to the null hypersurface and the resulting metric takes the form
\beq
    {\bf g} = -2\bck \bcl + 2\bcm \bar{\bcm},
\eeq

\noindent then the degenerate metric on $\mathcal{N}$ is
\beq 
    {\bf h} = 2 \bcm \bar{\bcm}
\eeq
\noindent where $\bcm$ is a complex-valued one-form. On $\mathcal{N}$ we may always include a real-valued $\bck$ so that $\bcm \wedge \bar{\bcm} \wedge \bck \neq 0$. Pulling back the Lorentz frame transformations in equation \eqref{grp:kfixed} yields the following  transformations
\beq
\begin{aligned}
    \bcm' &= e^{i\phi} \bcm, \\     
    \bck' &= A[\bck+z\bar{\bcm} + \bar{z} \bcm].
\end{aligned} \label{grp:CfrmlCarroll}
\eeq
\noindent This group is just the 1-dimensional extension of the Carroll group with parameter $A$. The original definition of a NHS can be restated as \cite{nurowski2000}:

\begin{definition}
    A null hypersurface structure (NHS) is a 3-dimensional manifold $\mathcal{N}$ equipped with an equivalence class of one-forms $(\bcm, \bck) \equiv (\bcm, \bar{\bcm}, \bck)$ such that
    \begin{itemize}
        \item $\bcm$ is complex-valued and $\bck$ is real-valued.
        \item $\bcm \wedge \bar{\bcm} \wedge \bck \neq 0$ at every point of $\mathcal{N}$.
        \item Given $(\bcm , \bck)$ and $(\bcm', \bck')$, they are equivalent if and only if they are related by the transformations \eqref{grp:CfrmlCarroll}.
    \end{itemize}
\end{definition}

This definition is somewhat more general than a pre-Carrollian structure \cite{blitz2023horizons} due to the transformation group admitting a real-valued parameter, $A$. However, by setting this parameter to one  and choosing a fundamental vector field ${\bf k}$, we recover a Carrollian structure. In a similar manner as Carrollian manifolds, we may also choose a connection for the NHS:

\begin{definition}
    A null hypersurface manifold is a NHS equipped with an affine connection.
\end{definition}

As mentioned in the introduction, we are interested in the characterization of embedded null hypersurfaces in spacetime. Following the classification in \cite{nurowski2000}, we could consider the invariants arising from the exterior differential systems that arise by pulling back to $\mathcal{N}$. Pulling back to the null hypersurface, 
the first Cartan structure equations become:
\begin{equation}
\begin{aligned}
	&\ext \bcm = - \bgam_{21} \wedge \bcm -  \bgam_{24} \wedge \bck, \\
	& 0 = \bgam_{41} \wedge \bcm + \bgam_{42} \wedge \bar{\bcm},\\
	&\ext \bck = \bgam_{31} \wedge \bcm + \bgam_{32} \wedge \bar{\bcm} + \bgam_{34} \wedge \bck.  
\end{aligned} \label{NullCartan1}
\end{equation}
\vspace{ 2 mm}

A similar pullback of the second Cartan structure equations is
\begin{widetext}
\begin{equation}
\begin{aligned}
	\ext\bgam_{23} &= (\bgam_{12} + \bgam_{34}) \wedge \bgam_{23}  - \Psi_3  \bcm \wedge \bar{\bcm}  + (\Psi_2 + \frac{1}{12} R) \bck \wedge \bcm  + \frac{1}{2} S_{32} \bcm \wedge \bar{\bcm}  + \frac{1}{2}S_{22} \bck \wedge \bar{\bcm}, \\
	\ext \bgam_{14} &= \bgam_{14} \wedge (\bgam_{12} + \bgam_{34}) + \Psi_1  \bcm \wedge \bar{\bcm} - \Psi_0 \bck \wedge \bcm - \frac{1}{2} S_{41} \bcm \wedge \bar{\bcm} - \frac{1}{2}S_{44} \bck	\wedge \bar{\bcm}, \\
		 \frac{1}{2}(d\bgam_{12} +\ext\bgam_{34}) &= \bgam_{23} \wedge \bgam_{14}  + (\Psi_2 - \frac{1}{24}R)  \bcm \wedge \bar{\bcm}  - \Psi_1 \bck \wedge \bcm  - \frac{1}{4}(S_{12} + S_{34}) \bcm \wedge \bar{\bcm} -\frac{1}{2}S_{42} \bck \wedge \bar{\bcm}. 
\end{aligned} \label{NullCartan2}
\end{equation}
\end{widetext}

This approach is independent of the choice of connection on $\mathcal{N}$ and explicitly classifies null hypersurfaces in terms of invariantly-defined quantities. In particular, the quantities $\rho$ and $\sigma$ can be used to broadly classify null hypersurfaces into four distinct cases~\cite{nurowski2000}. We will see that this approach picks out a preferred connection for a given hypersurface embedded in the spacetime. 

We note that the equations in \eqref{CartanStr2} could be expanded to give the NP field equations. Similarly, by exterior differentiating \eqref{CartanStr2} we may recover the Bianchi identities. An advantage of deriving the equations in this manner is that one may explicitly derive identities that are helpful when pulled back to the null hypersurface.

\section{Field equations for null hypersurfaces in Einstein solutions}

In the mathematics literature, a spacetime is said to be Einstein when it satisfies the vacuum Einstein field equations. Explicitly, this means that
$$R_{ij} - \tfrac{1}{2} g_{ij} R + \Lambda g_{ij}= 0\,,$$
for some constant $\Lambda$. Taking the trace of this expression we find that the Ricci scalar is a constant, $R = 4 \Lambda$. This family of spacetimes includes many of physical interest, including but not limited to the prototypical black hole solutions. Furthermore, by setting $\Lambda=0$, we recover all asymptotically flat vacuum solutions, including the Schwarzschild black hole and the Kerr black hole.

Recall from Section~\ref{sec2} that the pullback of the first Cartan structure equations~\eqref{NullCartan1} reduces to two equations and a constraint equation. Pulling back the definition of the connection coefficients in equation~\eqref{Bulk1forms}, we have that


\begin{equation}
\begin{aligned}
	& {\bf \bgam}_{41} = \sigma \bcm + \rho \bar{\bcm}, \\
	& \bgam_{23} = \mu \bcm + \lambda \bar{\bcm} + \pi \bck, \\
	& \frac12 (\bgam_{12} + \bgam_{34}) = -\beta \bcm - \alpha \bar{\bcm} - \epsilon \bck. 
\end{aligned} \label{Null1forms}
\end{equation}

Similarly, imposing the Einstein condition on the pullback of the second Cartan structure equations~\eqref{NullCartan2} yields
\begin{widetext}
\begin{equation}
\begin{aligned}
	\ext\bgam_{23} &= (\bgam_{12} + \bgam_{34}) \wedge \bgam_{23}  - \Psi_3  \bcm \wedge \bar{\bcm}  + (\Psi_2 + \frac{1}{12} R) \bck \wedge \bcm , \\
	\ext \bgam_{14} &= \bgam_{14} \wedge (\bgam_{12} + \bgam_{34}) + \Psi_1  \bcm \wedge \bar{\bcm} - \Psi_0 \bck \wedge \bcm  \\
		 \frac{1}{2}(d\bgam_{12} +\ext\bgam_{34}) &= \bgam_{23} \wedge \bgam_{14}  + (\Psi_2 - \frac{1}{24}R)  \bcm \wedge \bar{\bcm}  - \Psi_1 \bck \wedge \bcm. 
\end{aligned} \label{NullCartan2Einstein}
\end{equation}
\end{widetext}

We remind the reader that we have chosen to write the pulled back connection 1-forms with indices lowered. In particular, we have the following relationships 
\beq
\begin{aligned}
    \bgam_{23} &= - \bgam^4_{~2}, \\
    \bgam_{41} &= - \bgam^2_{~4}, \\
    \bgam_{12} &= \bgam^2_{~2},\\
    \bgam_{34} &= \bgam^4_{~4}.
\end{aligned}
\eeq
\noindent along with their complex conjugates. The metric compatibility condition in the spacetimes implies that $\Gamma_{11a} = \Gamma_{22a} = 0,\quad a=1,2,4$, which ensures that the above connection one-forms on the null hypersurface contain all the information necessary to specify a connection.

Combining equations \eqref{NullCartan1}, \eqref{Null1forms}, and the second equation in 
\eqref{NullCartan2Einstein} we may write $\Psi_0$ as 
\beq 
    \Psi_0 = \bk(\sigma) + (\bar{\epsilon} - 3 \epsilon - \rho) \sigma \label{Psi0} 
\eeq
\noindent From the coarse classification of NHSs in \cite{nurowski2000} using $\rho$ and $\sigma$ there are four cases: 
\begin{enumerate}
    \item $\sigma=\rho =0, $
    \item $\sigma =0, \rho \neq 0,$
    \item $\sigma \neq 0, \rho =0 $, 
    \item $\sigma \neq 0, \rho \neq 0$.    
\end{enumerate}
For Einstein spacetimes, the third case is necessarily excluded due to the null Raychaudhuri equation \cite[Chapter 5]{kramer}:
\beq 
	-{\bf k}(\rho) + \rho^2 + |\sigma|^2 =  -\frac12 S_{44}. \label{Raychaudhuri}
\eeq 
\noindent In the case of an Einstein spacetime with $\rho = 0$, this equation implies that $\sigma = 0$. In this paper we will focus on the null hypersurfaces where $\sigma = 0$ as these hypersurfaces have physical relevance as shear-free hypersurfaces. Thus, imposing that $\sigma = 0$ in equation~(\ref{Psi0}), it follows that $\Psi_0 = 0$. This implies that the Weyl tensor is at least of Petrov type I when pulled back to the null hypersurface. This further implies that, on the hypersurface, $\bck$ is aligned with the Weyl tensor. Following the analysis by Szekeres \cite{szekeres1965gravitational}, this implies that there are no incoming transverse gravitational waves affecting the hypersurface.

We note that $\Psi_3 $ and $\Psi_4$ will not explicitly appear in the curvature tensor of the corresponding Carrollian geometry, as they are coefficients of terms that contain $\bell$. Despite this, $\Psi_3$ appears as a coefficient of $\bcm \wedge \bar{\bcm}$ in the second Cartan structure equations and hence will influence the behaviour of the resulting Carrollian geometry. This confirms our intuition that outgoing longitudinal waves can affect a null hypersurface's intrinsic geometry. The absence of $\Psi_4$ in the equations implies that Carrollian hypersurface is not affected by potential outgoing transverse gravitational waves. We may treat the pullbacks of $\Psi_i$, $i=1,2,3$ onto the hypersurfaces as extrinsic data that give information about the interaction of the hypersurface with the bulk spacetime. Continuing with the analysis in \cite{szekeres1965gravitational}, $\Psi_2$ describes the Coulomb-like part of the gravitational field while $\Psi_1$ and $\Psi_3$ describe incoming or outgoing longitudinal gravitation waves in the spacetime. The pullback of these quantities give insight into dynamics of the hypersurface. 

We may explicitly integrate the field equations for any shear-free NHS in an Einstein spacetime to determine the coframe and connection one-forms. In what follows, we we will employ index notation for the connection coefficients instead of the quasi-NP formalism used in \cite{nurowski2000}. The following index convention will be used for covariant differentiation of basis one-forms:
\beq 
 \nabla_a \theta^b = \Gamma^b_{~ca} \theta^c,\quad a,b,c = 1,2,4.
\eeq
\noindent This ordering is consistent with the ordering used in the Cartan structure equations for the original spacetime. We present the solutions to the two cases for shear-free NHSs as two propositions and include the proofs in the following subsections.

\begin{proposition} \label{Prop:NEHCarroll}
    Any solution to the Einstein structure equations \eqref{NullCartan2Einstein} on a null hypersurface $\mathcal{N}$ satisfying $\rho = \sigma = 0$ is gauge equivalent to

    \beq 
    \begin{aligned}
        \bcm = e^{w}\ext \zeta, \quad \bck =\ext r + (X + r \bar{w}_{,\zeta})\ext \zeta + (\bar{X} + r w_{,\bar{\zeta}})\ext \bar{\zeta},
    \end{aligned} \label{Cfrm:Prop1}
    \eeq
\noindent with the following connection coefficients and their complex conjugates:   
\begin{widetext}
    \begin{equation}
    \begin{aligned}
       & \Gamma^2_{~41}=0,~\Gamma^2_{~42}=0, \Gamma^2_{~44} = 0,\\
       & \Gamma^2_{~21} = - \bar{w}_{,\zeta} e^{-w},\quad \Gamma^2_{~22} = w_{,\bar{\zeta}} e^{-\bar{w}}, \quad \Gamma^2_{~24}=0, \\
        &\Gamma^4_{~41}  = - \bar{w}_{,\zeta} e^{-w}, \quad \Gamma^4_{~42} = -w_{,\bar{\zeta}} e^{-\bar{w}}, \quad \Gamma^4_{~44} = 0, \\
        & \Gamma^4_{~21} = -\left[ \left(\frac{R}{8} e^w + \bar{w}_{,\zeta ,\bar{\zeta}} e^{-\bar{w}} \right) r + \left( S + \frac12 \left( X_{,\bar{\zeta}} - \bar{X}_{,\zeta}+w_{,\bar{\zeta}} X - \bar{w}_{,\zeta} \bar{X} \right) \right) e^{-\bar{w}} \right] e^{-w}, \\
        &\Gamma^4_{~22} = 0,\quad \Gamma^4_{~24}=0.
    \end{aligned} \label{Con:Prop1}
    \end{equation}
\end{widetext}

\noindent Here, $w = w(\zeta, \bar{\zeta}), X=X(\zeta, \bar{\zeta})$ are arbitrary complex-valued functions, and $S = S(\zeta, \bar{\zeta})$ is an arbitrary real-valued function. 
\vspace{3 mm}

\noindent The pulled back NP Weyl scalars are
\begin{widetext}
\beq
\begin{aligned}
    & \Psi_0 = 0\\ 
    & \Psi_1 = 0 \\    
    & \Psi_2 = \frac{R}{24} + \bar{w}_{,\zeta \bar{\zeta}} e^{-(w + \bar{w})} \\
    & \Psi_3 = e^{-(w + \bar{w})} \left[ \left( \left( \frac{R}{8} e^w + \bar{w}_{,\zeta \bar{\zeta}} \right) r +\frac12 e^{-\bar{w}} (2S + X_{,\bar{\zeta}} - \bar{X}_{,\zeta} + w_{,\bar{\zeta}} X - \bar{w}_{,\zeta} \bar{X}) \right)_{,\bar{\zeta}} - \left( \frac{R}{8} e^w + \bar{w}_{,\zeta \bar{\zeta}} e^{-\bar{w}} \right) (\bar{X}+rw_{,\bar{\zeta}})\right].
\end{aligned} \nonumber
\eeq
\end{widetext}
\end{proposition}

\begin{proposition} \label{Prop:NotNEHCarroll}
    Any solution to the Einstein structure equations \eqref{NullCartan2Einstein} on a null hypersurface $\mathcal{N}$, satisfying $\rho \neq 0$ and  $\sigma = 0$ is gauge equivalent to
    \beq
    \begin{aligned}
        \bcm = \frac{d\zeta}{\rho}, \quad \bck = - \ext \left(\frac{1}{\rho}\right)+\frac12 \rho(x\ext \zeta + \bar{x}d\bar{\zeta}),
    \end{aligned} \label{Crfm:Prop2}
    \eeq

    \noindent with the following connection coefficients and their complex conjugates:
    \begin{widetext}
    \begin{eqnarray*}
    \begin{aligned}
        &\Gamma^2_{~41} = 0,\quad \Gamma^2_{~42} = - \rho, \quad \Gamma^2_{~44} = 0, \\
        & \Gamma^2_{~21} = -\frac12 x\rho^3,\quad \Gamma^2_{~22} =  \frac12 \bar{x} \rho^3, \quad \Gamma^2_{~24} = 0, \\
        & \Gamma^4_{~41} = -\frac12 x\rho^3,\quad \Gamma^4_{~42} =  - \frac12 \bar{x} \rho^3, \quad \Gamma^4_{~44} = 0, \\
        & \Gamma^4_{~21} = -\rho \left[ -\frac16 \rho^4 |x|^2 - \frac12 \rho^2 x_{,\bar{\zeta}} + \rho a - \frac{R}{24 \rho^2} \right],\\
        &\Gamma^4_{~22} = -\rho q, \Gamma^4_{~24} = 0. 
    \end{aligned} \label{Con:Prop2}
    \end{eqnarray*}
    \end{widetext}
    \noindent Here, $x = x(\zeta, \bar{\zeta}), q=q(\zeta, \bar{\zeta})$ are arbitrary complex-valued functions, and $a = a(\zeta, \bar{\zeta})$ is an arbitrary real-valued function. 
\vspace{3 mm} 

    The pulled back NP Weyl scalars are 

\beq 
\begin{aligned}
    \Psi_0 &= 0, \\
    \Psi_1 &= x \rho^4,\\
    \Psi_2 &= a\rho^3 + \frac23 \rho^6 x \bar{x} + \rho^4 x_{,\bar{\zeta}} \\
    \Psi_3 &= -\frac13 \rho^8 x \bar{x}^2 - \rho^6 \left(\frac{(\bar{x} x)_{,\bar{\zeta}}}{6} + \frac12 \bar{x} x_{,\bar{\zeta}} \right) - \frac12 \rho^5 \bar{x} a + \rho^4 \left( \frac12 x_{,\bar{\zeta} \bar{\zeta}} + q x \right) - \rho^4 a_{,\bar{\zeta}} + \rho^2 \left( q_{,\zeta} - \frac{R}{24} \bar{x} \right).  
\end{aligned}
\eeq

\end{proposition} 

We note that the cosmological constant from the original four-dimensional spacetime does not affect the coframe. Instead it affects the connection coefficients and the pulled back NP Weyl scalars $\Psi_2$ and $\Psi_3$. In particular, this implies that the degenerate metric for asymptotically flat spacetimes is identical to those for non-trivial Einstein spacetimes (like Kerr-de Sitter) in the given coordinate systems.

\subsection{Proof of Proposition \ref{Prop:NEHCarroll}}

Setting $\rho = \sigma = 0$ in $\bgam_{41}$ in \eqref{Null1forms}, this implies that $\bgam_{41} = 0.$ Imposing this condition, the second equation in \eqref{NullCartan1} vanishes automatically, while the first equation in \eqref{NullCartan1}, implies that \beq 
   \ext \bcm \wedge \bcm = - \bgam_{21} \wedge \bcm \wedge \bcm  = 0
\eeq
\noindent and so there exists (see for example \cite{PawelNotes2011}) two complex functions on $\mathcal{N}$,  $w$ and $\zeta$ such that 
\beq
    \bcm = e^w\ext\zeta \label{Case1:M}\,.
\eeq
\noindent Since $\bcm \wedge \bar{\bcm} \neq 0$, it follows that $e^{2Re(w)}\ext \zeta \wedge\ext \bar{\zeta} \neq 0$.

Imposing $\bgam_{41} = 0$ on the equations in \eqref{NullCartan2Einstein}, the second equation implies that $\Psi_1 =0$ while the third equation implies 
\beq 
   \ext(\bgam_{12} + \bgam_{34}) \wedge \bcm =\ext(\bgam_{12} + \bgam_{34}) \wedge ( e^w \ext\zeta) = 0. 
\eeq
\noindent This implies the existence of two complex functions $\eta$ and $h$ such that $\bgam_{12} + \bgam_{34} =\ext\eta + h\ext \zeta$ \cite{PawelNotes2011}. Using the frame freedom \eqref{grp:CfrmlCarroll} we can fix the parameters $A$ and $\phi$ so that 
\beq 
    \bgam_{12}+\bgam_{34} = h\ext  \zeta \,.\label{Case1:Gam12Gam34}
\eeq
\noindent Now inserting this 1-form into the third equation in \eqref{NullCartan2Einstein}, we find 
\beq 
    \ext h \wedge\ext\zeta = 2 e^{w+\bar{w}}\left(\Psi_2 - \frac{R}{24} \right)\ext \zeta \wedge\ext\bar{\zeta}\,.
\eeq\noindent It follows that $h$ is a function of $\zeta$ and $\bar{\zeta}$ only, and further that 
\beq
    h_{,\bar{\zeta}} = 2 \left(\frac{R}{24} - \Psi_2 \right)e^{w+\bar{w}}\,. \label{Case1:heqn}
\eeq
Using the fact that $\bar{\bgam}_{34} =\bgam_{34}$   and $\bar{\bgam}_{12} =-\bgam_{12}$, it follows  from the imaginary part of equation \eqref{Case1:Gam12Gam34}, that $\bgam_{12} = \frac12 (h\ext \zeta - \bar{h}d\bar{\zeta})$. We may plug this into the first equation of \eqref{NullCartan1} to find that $w=w(\zeta, \bar{\zeta})$ only and 
\beq 
    h = - 2\bar{w}_{,\zeta}.
\eeq
\noindent From \eqref{Case1:heqn}, and the above identity, we may solve for $\Psi_2$
\beq
    \Psi_2 = \frac{R}{24} + \bar{w}_{,\zeta \bar{\zeta}} e^{-(w+\bar{w})}.
\eeq
\noindent Furthermore, the first equation in \eqref{NullCartan2Einstein} gives
\beq 
   \ext \bgam_{23} = - 2 \bar{w}_{,\zeta}\ext \zeta \wedge \bgam_{23} - \Psi_3 e^{w+\bar{w}}\ext \zeta \wedge\ext \bar{\zeta} + \left( \frac{R}{8} + \bar{w}_{,\zeta \bar{\zeta}} e^{-(w+\bar{w})}\right) e^w \bck \wedge\ext\zeta. \label{Case1: dGam23}
\eeq
\noindent Thus, we find that $\ext \bgam_{23} \wedge \ext \zeta =0$, and there must exist two complex-valued functions on $\mathcal{N}$, $v$ and $q$ such that $\bgam_{23} = \ext v + q \ext \zeta$ \cite{PawelNotes2011}. Fixing the $z$-parameter in \eqref{grp:CfrmlCarroll}, this can be set to be
\beq 
    \bgam_{23} = q\ext \zeta. \label{Case1:Gam23}
\eeq
\noindent We note that all of the frame transformation parameters have been fixed now. 

Inserting equation \eqref{Case1:Gam23} into equation \eqref{Case1: dGam23} we get
\beq 
    \ext q = \Psi_3 e^{w+\bar{w}}\ext \bar{\zeta} + \left( \frac{R}{8} + \bar{w}_{,\zeta \bar{\zeta}} e^{-(w+\bar{w})}\right)e^w \bck + s\ext \zeta, \label{Case1:dq}
\eeq
\noindent where $s$ is an arbitrary complex-valued function. Using this in the third equation in \eqref{NullCartan1}, we have
\beq 
   \ext \bck = (\bar{q}e^w - q e^{\bar{w}})\ext \zeta \wedge\ext\bar{\zeta} + \bck \wedge [\bar{w}_{,\zeta}\ext \zeta + w_{,\bar{\zeta}}\ext \bar{\zeta}] \,.\label{Case1:dK}
\eeq

Supplementing $\zeta$ and $\bar{\zeta}$ by a real variable $r$ to produce a coordinate system $(r, \zeta, \bar{\zeta})$ on $\mathcal{N}$ such that $\bk = \partial_r$, we may express $\bck$ as
\beq
    \bck =\ext r + Y\ext \zeta + \bar{Y}\ext\bar{\zeta},
\eeq
\noindent where $Y$ is a complex-valued function of $(r, \zeta, \bar{\zeta})$. Plugging this expression for $\bck$ into equation \eqref{Case1:dK} we find
\beq
    dY = q e^{\bar{w}}\ext \bar{\zeta} + \bar{w}_{,\zeta}[dr + \bar{Y}d\bar{\zeta}] + ud\zeta + t\ext \bar{\zeta}, \label{Case1:dY}
\eeq
\noindent where $t$ and $u$ are real-valued and complex-valued functions on $\mathcal{N}$, respectively. 

This equation is equivalent to a system of partial differential equations (PDEs):
\beq
\begin{aligned}
    Y_{,r} &= \bar{w}_{,\zeta},\\
    Y_{,\bar{\zeta}} &= \bar{w}_{,\zeta} \bar{Y}+qe^{\bar{w}} + t, \\
    Y_{,\zeta} &= u\,.
\end{aligned} \label{Case1:Ypde}
\eeq
\noindent Integrating the first equation of \eqref{Case1:Ypde}, we have
\beq
    Y = r\bar{w}_{,\zeta} + X,
\eeq
\noindent where $X = X(\zeta, \bar{\zeta})$ is an arbitrary complex-valued function. Using this expression for $Y$ in the second equation of \eqref{Case1:Ypde} and taking the imaginary part gives
\beq
    r(\bar{w}_{,\zeta \bar{\zeta}}-w_{,\zeta \bar{\zeta}})+ X_{,\bar{\zeta}} - \bar{X}_{,\zeta} = \bar{w}_{,\zeta} \bar{X} - w_{,\bar{\zeta}} X + qe^{\bar{w}}-\bar{q}e^w.  \label{Case1:ugly}
\eeq

If we transcribe equation \eqref{Case1:dq} into a system of PDEs, we have
\beq
\begin{aligned}
    q_{,r} &= \frac{R}{8} e^w + \bar{w}_{,\zeta \bar{\zeta}} e^{-\bar{w}},\\
    q_{,\bar{\zeta}} &= \Psi_3 e^{w+\bar{w}} + \left( \frac{R}{8} e^w + \bar{w}_{,\zeta \bar{\zeta}} e^{-\bar{w}}\right) \bar{Y}.
\end{aligned} \label{Case1:qpde}
\eeq

\noindent Integrating the first equation in \eqref{Case1:qpde} we find
\beq 
    q = \left[ \frac{R}{8} e^w + \bar{w}_{,\zeta \bar{\zeta}} e^{-\bar{w}} \right] r + Q e^{-\bar{w}} 
\eeq
\noindent where $Q = Q(\zeta, \bar{\zeta})$ is an arbitrary complex-valued function. Note this integration required that $R$ is constant, which followed from the Einstein condition. Inserting this into equation \eqref{Case1:ugly} we find a constraint on the imaginary part of $Q$:
\beq 
Q-\bar{Q} = X_{,\bar{\zeta}} - \bar{X}_{,\zeta} - \bar{w}_{,\zeta} \bar{X} + w_{,\bar{\zeta}} X. 
\eeq
\noindent Writing $Q + \bar{Q} = S$, this is an arbitrary  real-valued function of $\zeta$ and $\bar{\zeta}$. Using the second equation in \eqref{Case1:qpde} we may solve for $\Psi_3$. We have now fully integrated out the differential conditions for $\bcm, \bck$ and the connection 1-forms on the null hypersurface $\mathcal{N}$, the results of which are summarized in Proposition \ref{Prop:NEHCarroll}.



\subsection{Proof of Proposition \ref{Prop:NotNEHCarroll}}

Setting $\sigma = 0$ in $\bgam_{41}$ in \eqref{Null1forms}, this implies that 
\beq 
    \bgam_{41} = \rho \bar{\bcm},\quad \bar{\rho} = \rho. \label{Case2:Gam41a}
\eeq
\noindent Furthermore, from equation \eqref{Psi0}, it follows that 
\beq
    \Psi_0 = 0.
\eeq
\noindent The first equation in \eqref{NullCartan1} with equation~\eqref{Case2:Gam41a} implies that $\ext \bgam_{41} \wedge \bgam_{41} =0$ which in turn implies the existence of two complex-valued functions $f$ and $\zeta$ such that $\bgam_{14} = \bar{f}\ext \bar{\zeta}$. Using the $A$ and $\phi$ parameters in the frame transformation group \eqref{grp:CfrmlCarroll} we may set $\bar{f} =-1$ giving
\beq 
    \bgam_{41} = - \ext \bar{\zeta}. \label{Case2:Gam41b}
\eeq
Then equating \eqref{Case2:Gam41a} with \eqref{Case2:Gam41b} gives
\beq 
    \bcm = -\frac{\ext \zeta}{\rho} \label{Case2:M}.
\eeq

Using \eqref{Case2:Gam41b} in the second equation in \eqref{NullCartan2Einstein} we can express another connection 1-form:
\beq
    \bgam_{12} + \bgam_{34} =  \frac{\Psi_1}{\rho^2} \ext \zeta + h \bgam_{14},
\eeq

\noindent where $h$ is an arbitrary complex-valued function on $\mathcal{N}$. We now may fix the remaining $z$ parameter in the frame transformation group \eqref{grp:CfrmlCarroll} to set $h=0$ and yield
\beq
    \bgam_{12}+ \bgam_{34} =  \frac{\Psi_1}{\rho^2} \ext \zeta. \label{Case2:Gam12Gam34}
\eeq


Taking $\bcm$, $\bgam_{12}$ and $\bgam_{24}$ in equation \eqref{Case2:M}, the real part of equation \eqref{Case2:Gam12Gam34}, and the complex conjugate of \eqref{Case2:Gam41a}, respectively, then the first equation in \eqref{NullCartan1} we can solve explicitely for the remaining coframe element:
\beq
    \bck = - \ext \left( \frac{1}{\rho} \right) - \frac12 \frac{\Psi_1}{\rho^3} \ext \zeta -\frac12 \frac{\bar{\Psi}_1}{\rho^3} \ext \bar{\zeta}. \label{Case2:K}
\eeq
\noindent In addition, since $\rho$ is real-valued, we can adapt $(\rho, \zeta, \bar{\zeta})$ as a coordinate system on $\mathcal{N}$.

Substituting $\bgam_{12} + \bgam_{34}$ in \eqref{Case2:Gam12Gam34} into the third equation of \eqref{NullCartan2Einstein}, it follows that there exists three complex-valued functions on $\mathcal{N}$: $p,q$ and $s$ such that 
\beq 
    \bgam_{23} = \frac{p}{\rho^2} \ext \zeta + q \ext \bar{\zeta}, \label{Case23}
\eeq
\noindent along with
\beq
    \rho^4 \ext \left( \frac{\Psi_1}{\rho^4} \right) + \left( 2 p - \frac{R}{12} + 2\Psi_2 + \frac{\Psi_1 \bar{\Psi}_1}{\rho^2}\right) \ext \bar{\zeta} + 2 s \ext  \zeta = 0 \label{Case2:Psi1DE}
\eeq
\noindent Notice that this second equation implies that 
\beq 
    \rho^4 \ext \left( \frac{\Psi_1}{\rho^4} \right) \wedge \ext \zeta \wedge \ext \bar{\zeta} =0. \label{BetterPsi1DE}
\eeq

There are two cases, depending on whether $\Psi_1$ vanishes or not. We will consider the case when $\Psi_1 \neq 0$ first. Then equation \eqref{BetterPsi1DE} implies that there is a function $x = x(\zeta, \bar{\zeta})$ such that 
\beq 
    \Psi_1 = \rho^4 x. \label{Case2:Psi1}
\eeq
\noindent Setting this into \eqref{Case2:Psi1DE} allows for $p$ and $s$ to be solved for algebraically:
\beq 
    \begin{aligned}
        p &= - \Psi_2 + \frac{R}{24} - \frac12 \rho^4 (x_{,\bar{\zeta}} + \rho^2 x \bar{x}), \\
        s &= -\frac12 \rho^4 x_{,\zeta}.
   \end{aligned} \label{Case2:p}
\eeq

With the function, $p$, set into the third equation in \eqref{NullCartan1}, it can only be satisfied if
\beq 
    \bar{\Psi}_2 - \Psi_2 = \rho^4 ( \bar{x}_{,\zeta} - x_{,\bar{\zeta}}). \label{case2:ImPsi2}
\eeq
\noindent The remaining equation which must be examined is the first equation of \eqref{NullCartan2Einstein}. Computing the $\ext \rho \wedge \ext \zeta$ and $\ext \rho \wedge \ext \bar{\zeta}$ terms:
\beq
\begin{aligned}
    q_{,\rho} & = 0\,, \\
    \Psi_2 &= a(\zeta, \bar{\zeta}) \rho^3 + \frac23 \rho^6 x \bar{x} + \rho^4 x_{,\bar{\zeta}}\,,
\end{aligned}
\eeq

\noindent where $a$ is an arbitrary function of integration dependent on $\zeta$ and $\bar{\zeta}$. The first equation implies that $q$ must be independent of $\rho$, i.e., $q = q(\zeta, \bar{\zeta})$. The second equation is compatible with equation \eqref{case2:ImPsi2} if and only if $a = \bar{a}$, i.e., $a$ is a real-valued function. 

Finally, the $\ext \zeta \wedge \ext \bar{\zeta}$ term in the first equation in \eqref{NullCartan2Einstein} gives an algebraic expression for $\Psi_3$
\beq 
    \Psi_3 = -\frac13 \rho^8 x \bar{x}^2 - \rho^6 \left(\frac{(\bar{x} x)_{,\bar{\zeta}}}{6} + \frac12 \bar{x} x_{,\bar{\zeta}} \right) - \frac12 \rho^5 \bar{x} a + \rho^4 \left( \frac12 x_{,\bar{\zeta} \bar{\zeta}} + q x \right) - \rho^4 a_{,\bar{\zeta}} + \rho^2 \left( q_{,\zeta} - \frac{R}{24} \bar{x} \right).
\eeq
\noindent This solves the equations \eqref{NullCartan1} and \eqref{NullCartan2Einstein} in the case that $\Psi_1 \neq 0$. We note that solutions to these equations when $\Psi_1 = 0$ can be formally obtained by setting $x=0$. The results of this analysis are summarized in Proposition \ref{Prop:NotNEHCarroll}.



\section{Shear-free Null hypersurface structures as Carrollian manifolds}

For $\mathcal{N}$ an NHS embedded in a shear-free Einstein spacetime, there is automatically a degenerate metric, ${\bf h} = 2 \bcm \bar{\bcm}$. Following from Proposition \ref{Prop:NEHCarroll} and \ref{Prop:NotNEHCarroll}, where the coframe basis and connection have been determined for all possible shear-free null hypersurfaces, we may view $\bck$ as an Ehresmann connection. Then, by computing the dual of the coframe, $(\bm, \bar{\bm}, \bk)$, we may recover a vector-field that exists in the tangent space of $\mathcal{N}$ that we use as the fundamental vector field:

\beq
\begin{aligned}
    \text{Proposition \ref{Prop:NEHCarroll}} (\sigma = \rho = 0)&:~ \bm = -(X+r\bar{w}_{,\zeta}) \partial_r + e^{w} \partial_{\zeta}, \quad \bar{\bm} =  -(\bar{X}+r w_{,\bar{\zeta}})\partial_r + e^{\bar{w}} \partial_{\bar{\zeta}}, \quad \bk = \partial_r \\
    \text{Proposition \ref{Prop:NotNEHCarroll} } (\sigma = 0, \rho \neq 0)&:~ \bm = -\frac{x \rho^4}{2} \partial_\rho + \rho \partial_{\zeta}, \quad \bar{\bm} = -\frac{\bar{x} \rho^4}{2} \partial_\rho + \rho \partial_{\bar{\zeta}}, \quad \bk = \rho^2 \partial_\rho.
\end{aligned} \label{Props:frame}
\eeq

\noindent Thus, the NHS may be treated as a Carrollian structure $(\mathcal{N}, {\bf h}, \bk)$ through a  restriction of the frame transformation group to the Carrollian group by fixing the $A$ parameter in equation \eqref{grp:kfixed}. Furthermore, this Carrollian structure is equipped with a distinguished affine connection, and so it is actually a Carrollian manifold.

We have seen that both the coframe and the connection for a NHS can be determined by solving the exterior differential system arising from the pullback of the Cartan structure equations. This was done in an invariant manner using the frame transformation group in equation \eqref{grp:CfrmlCarroll} to simultaneously determine the solutions and invariantly classify them by normalizing tensor quantities in the exterior differential system independent of any coordinate system. While the frame parameters in equation \eqref{grp:CfrmlCarroll} have been fixed, once the solutions have been found, one may transform away from the frame to achieve a different normalization. While we have required that the connection is torsion-free and metric-compatible in the bulk, these conditions do not necessarily hold for the induced connection on the NHS. 

Taking the degenerate metric, it is a straightforward calculation to determine the nonmetricity tensor, ${\bf Q} = \nabla {\bf h}$,

\beq
\begin{aligned}
    \sigma = \rho = 0 &:~ {\bf Q} = 0, \\
    \sigma = 0, \rho \neq 0 &:~ {\bf Q} = \rho (\bcm \otimes \bck \otimes \bar{\bcm} + \bar{\bcm} \otimes \bck \otimes \bcm + \bck \otimes {\bf h} )
\end{aligned}
\eeq

Similarly, computing the intrinsic torsion of the connection for arbitrary vector-fields, ${\bf X}$ and ${\bf Y}$, in the tangent space of $\mathcal{N}$:
\beq 
    {\bf T}({\bf X}, {\bf Y}) = \nabla_{{\bf X}} {\bf Y} - \nabla_{{\bf Y}} {\bf X} - [{\bf X}, {\bf Y}], \label{IntrinsicTorsion}
\eeq
\noindent we find that in the case that $\sigma = \rho = 0$ for $\mathcal{N}$, the torsion tensor is then
\beq
\begin{aligned}
    {\bf T} &= 2 e^{w - \bar{w}} \left[(  \bar{w}_{,\zeta \bar{\zeta}} - w_{,\zeta \bar{\zeta}}  )r + X w_{,\bar{\zeta}}- \bar{X} \bar{w}_{,\zeta}-\bar{X}_{,\zeta}+X_{,\bar{\zeta}}   \right] \bk \otimes (\bcm \wedge \bar{\bcm}) \\ & \quad + 2 \bar{w}_{\zeta} e^{-w} \bk \otimes (\bck \wedge \bar{\bcm}) + 2 w_{\bar{\zeta}} e^{-\bar{w}} \bk \otimes (\bck \wedge \bcm). 
\end{aligned} \label{Case1:Torsion}
\eeq
\noindent while in the case where $\sigma = 0$ and $\rho \neq 0$ the torsion tensor is
\beq 
\begin{aligned}
    {\bf T} &=  \frac{\rho^3}{2}(\bar{x}_{,\zeta}-x_{,\bar{\zeta}}) \bk \otimes  (\bcm \wedge \bar{\bcm}) + 2 \rho^3 \bar{x} \bk \otimes  (\bck \wedge \bar{\bcm}) + \rho^3 x \bk \otimes (\bck \wedge \bcm). 
\end{aligned} \label{Case2:Torsion}
\eeq

With the data contained in Propositions \ref{Prop:NEHCarroll} and \ref{Prop:NotNEHCarroll}, we may also compute the curvature tensor for these Carrollian geometries. The arbitrary functions that are contained in the connection coefficients but not the frame or coframe appear within the components of the curvature tensor. We note further that while expressions similar to $\Psi_1$ and $\Psi_2$ appear in the components of the curvature tensor, $\Psi_3$ and $\Psi_4$ will not. To understand this, observe that a similar phenomenon occurs even for non-null hypersurfaces. Indeed, for any signature hypersurface and a geometrically-determined (co)frame along it, differentiation tangent to the hypersurface can only extract components of the bulk curvature tensor contracted with at most one copy of the vector field transverse to the hypersurface. As the constructed intrinsic covariant derivative on the Carrollian manifold is equal to the tangential Levi-Civita action modified by an extrinsic invariant that is first-order in derivatives on the bulk metric, the Carrollian curvature can only pick up components of the bulk curvature containing at most one copy of the transverse vector field. Because both $\Psi_3$ and $\Psi_4$ are curvature components involving two copies of the transverse vector field, the fact that we do not see these components in the Carrollian curvature is expected.

The connections we have found for the shear-free NHS embedded in Einstein spacetimes are new. Since $\nabla {\bf h} \neq 0$ these connections are not metric compatible   \cite{blitz2023horizons}. Of course, if ${\bf Q} \neq {\bf 0}$ one could introduce a shift tensor, $\Delta \bgam$, 
\beq
    \Delta \bgam_Q = \rho (\bcm \otimes \bk \otimes \bcm + \bar{\bcm} \otimes \bk \otimes \bar{\bcm})
\eeq

\noindent so that $\bgam' =\bgam + \Delta \bgam$ is a metric compatible connection. However, for this new connection, it is not possible to require that $\nabla \bk = 0$ or that $\nabla \bck' = 0$ without imposing conditions on the arbitrary functions, as would be required if the connection constructed here could be modified to produce the connection in~\cite{blitz2023horizons}. This is expected as the connection discussed there is intended to describe the intrinsic properties of embedded null hypersurfaces as Carrollian structures. 

However, using the idea of a shift tensor we may construct a new torsion-free connection on $\mathcal{N}$ that is not metric compatible. We remark that this shift tensor is not quite a contorsion tensor as this would require raising and lowering of indices which is not permitted with a degenerate metric. For any NHS where $\sigma = \rho = 0$, we may use 
\beq
\begin{aligned}
    \Delta \bgam_T & = -2 e^{w - \bar{w}} \left[(  \bar{w}_{,\zeta \bar{\zeta}} - w_{,\zeta \bar{\zeta}}  )r + X w_{,\bar{\zeta}}- \bar{X} \bar{w}_{,\zeta}-\bar{X}_{,\zeta}+X_{,\bar{\zeta}}   \right] \bk \otimes \bcm \otimes \bar{\bcm} \\ 
    & \quad  -2 \bar{w}_{,\zeta} e^{-w} \bk \otimes \bcm \otimes \bck - 2 w_{,\bar{\zeta}}e^{-\bar{w}} \bk \otimes \bar{\bcm} \otimes \bck.
\end{aligned} \label{Shift:Prop1}
\eeq
\noindent while for the case where $\sigma = 0$ but $\rho \neq 0$ we may use
\beq
\begin{aligned}
        \Delta \bgam_T &= -\frac{\rho^3}{2}(\bar{x}_{,\zeta}-x_{,\bar{\zeta}}) \bk \otimes  \bcm \otimes \bar{\bcm} - 2 \rho^3 \bar{x} \bk \otimes  \bck \otimes \bar{\bcm} - \rho^3 x \bk \otimes \bck \otimes \bcm.
\end{aligned} \label{Shift:Prop2}
\eeq

Using $\bgam' = \bgam + \Delta \bgam_T$ we now have a torsion-free connection that still contains the arbitrary functions that arise from extrinsic information. This torsion-free connection is reminiscent of another choice of connection that is increasingly popular. Indeed, in the physics literature, this particular choice (summarized well in~\cite{freidel2024}) is also a connection that both contains extrinsic embedding data and is torsion-free (but is not metric compatible). 

In that context, this choice is made because it allows one to encode the bulk gravitational effects on the null hypersurface as aspects of the intrinsic geometry of the hypersurface, such as incoming gravitational radiation. The desired connection is determined by several constraints. First, one demands that the identity operator on the tangent bundle of the ruled Carrollian geometry is preserved under covariant differentiation and that the connection is symmetric. Then, for any two spacetime vector-fields $\bf{X}$ and $\bf{Y}$ tangent everywhere to $\mathcal{N}$, the action of the rigged connection~\cite{mars2023,mars2023-2} is uniquely fixed by the relation $\bar{\nabla}_{\bf X} {\bf Y} := (\nabla_{\bf X} {\bf Y})_{||}|_{\mathcal{N}}$, where the parallel part is determined by the projection operator built from the Ehresmann connection and its dual (which here is the rigging vector field). 

The pullback and projection of the resulting tensors then complete the specification of the connection. In this way, extrinsic embedding data is built in to the connection at the cost of no longer being adapted to the Carrollian structure.

It is thus worthwhile to consider when the connection constructed in this manuscript can be shifted to the connection built from a rigged connection. By definition a symmetric connection is torsion-free so this condition is already satisfied. Taking the connection in Propositions \ref{Prop:NEHCarroll} and \ref{Prop:NotNEHCarroll} and shifting them by the shift tensor in equation \eqref{Shift:Prop1} and \eqref{Shift:Prop2} respectively, we may verify that the identity operator, 
\beq
    {\bf q} = \bm  \otimes \bcm + \bar{\bm} \otimes \bar{\bcm} + \bk \otimes \bck,
\eeq
\noindent satisfies $\bar{\nabla} {\bf q} = 0$. Thus, the shifted connections can be treated as a rigged connection and we may specify the action of the torsion-free connections on the fundamental vector-field and Ehresmann connection.

\begin{corollary} \label{Cor:NEH}
    For any null hypersurface embedded in an Einstein spacetime with $\sigma = \rho = 0$ we may always construct a Carrollian geometry admitting a rigged connection where the fundamental vector-field has vanishing inaffinity, $\bck_b \bk^a \bar{\nabla}_a \bk^b$. 
\end{corollary}

\begin{proof}
    We use the torsion-free connection $\bgam' =\bgam + \Delta \bgam $ with $\bgam$ defined in Proposition \ref{Prop:NEHCarroll} and $\Delta \bgam$ defined in equation \eqref{Shift:Prop1}. Computing the covariant derivative of the fundamental vector field, $\bk$, and the Ehresmann connection, $\bck$, we find:
    \beq 
    \begin{aligned}
        \bar{\nabla} \bk &= \bk \otimes \left( \bar{w}_{,\zeta} e^{-w} \bcm + w_{,\bar{\zeta}}e^{-\bar{w}} \bar{\bcm} \right), \\
        \bar{\nabla} \bck &= \Gamma^4_{~(12)} {\bf h}  + (\Gamma^4_{~[12]} + [\Delta \Gamma_T]^3_{~[12]}) (\bcm \wedge \bar{\bcm}) \\
        & \quad + 2 \left( \bar{w}_{,\zeta} e^{-w} \bcm + w_{,\bar{\zeta}}e^{-\bar{w}} \bar{\bcm} \right) \otimes \bck + \bck \otimes \left( \bar{w}_{,\zeta} e^{-w} \bcm + w_{,\bar{\zeta}}e^{-\bar{w}} \bar{\bcm} \right).
    \end{aligned}
    \eeq
\end{proof}

\begin{corollary} \label{Cor:NotNEH}
    For any null hypersurface embedded in an Einstein spacetime with $\sigma=0$ and $\rho \neq0$ we may always construct a Carrollian geometry admitting a rigged connection where the fundamental vector-field has vanishing inaffinity, $\bck_b \bk^a \bar{\nabla}_a \bk^b$.
\end{corollary}

\begin{proof}
    We use the torsion-free connection $\bgam' =\bgam + \Delta \bgam $ with $\bgam$ defined in Proposition \ref{Prop:NotNEHCarroll} and $\Delta \bgam$ defined in equation \eqref{Shift:Prop2}. Computing the covariant derivative of the fundamental vector field, $\bk$, and the Ehresmann connection, $\bck$, we find:
    \beq 
    \begin{aligned}
        \bar{\nabla} \bk &= - \rho ({\bf q} - \bk \otimes \bck) - \bk \otimes \left( \frac{\rho^3 x}{2} \bcm + \frac{\rho^3 \bar{x}}{2} \bar{\bcm}\right), \\
        \bar{\nabla} \bck &= \Gamma^4_{~(12)} {\bf h}  + (\Gamma^4_{~[12]} + [\Delta \Gamma_T]^3_{~[12]}) (\bcm \wedge \bar{\bcm}) + \rho q \bcm \otimes \bcm + \rho \bar{q} \bar{\bcm} \otimes \bar{\bcm}  \\
        & \quad + \bk \otimes \left( \frac{\rho^3 x}{2} \bcm + \frac{\rho^3 \bar{x}}{2} \bar{\bcm}\right) + \left( \rho^3 x \bcm + \rho^3 \bar{x} \bar{\bcm}\right) \otimes \bk.
    \end{aligned}
    \eeq
\end{proof}

If we permit a scaling of $\bk$ and $\bck$ to preserve $\bck(\bk)=1$, then employing the Leibniz rule for covariant differentiation gives the next result.

\begin{corollary} \label{cor:AchangeL}
    For any shear-free NHS embedded in an Einstein spacetime, we may construct a Carrollian geometry with non-vanishing inaffinity by setting $\bk' = A^{-1} \bk$ and $\bck' = A \bck$ where $A$ is a function satisfying $\bk(A) \neq 0$.
\end{corollary}

\begin{remark}
    When trying to do particle physics on a Carrollian manifold, often a rigged connection is used. However, in some special cases, the rigging vector that determines the connection is not canonically determined~\cite{Ashtekar2002}---in particular, for extremal horizons, this is the case. However, our construction, at least when $\sigma = 0$, still works for extremal horizons, and thus allows for the construction of a geometrically-determined connection. Then, if one were to use a rigged connection to construct an action on a Carrollian manifold, there exists an additional geometric tensor, namely $\Delta \bgam$ (describing the difference between the rigged connection and the unique connection constructed here), that can be used to build more complicated (and possibly more interesting from a dynamical perspective) actions.
\end{remark}

\section{Conclusions
}

By explicitly solving the exterior differential system that arises from the Cartan structure equations projected onto a null hypersurface, we have determined all possible shear-free null hypersurfaces in Einstein spacetimes as null hypersurface structures (NHS) equipped with a unique pairing of coframe basis and connection. The structure group for NHSs contains the Carrollian group as a subgroup and can be recovered by fixing one of the group parameters. Of particular note, all non-expanding horizons (NEHs) and hence weakly isolated horizons (WIH) \cite{Ashtekar2002} in Einstein spacetimes are determined explicitly in this paper as NHSs in Proposition \ref{Prop:NEHCarroll}. Less is known about the null hypersurfaces described in Proposition \ref{Prop:NotNEHCarroll}, however, examples of these surfaces arise as shear-free conformal Killing horizons \cite{sultana2004conformal}.

Motivated by this observation, we then proved that for such hypersurfaces there is a uniquely determined Carrollian manifold with a unique pair of Ehresmann connection and affine connection. As such, this resolves a significant roadblock in discussing Carrollian geometry: how 
 one geometrically fixes the Ehresmann connection. Indeed, previous work by Ashtekar et. al.~\cite{Ashtekar2002} showed that an Ehresmann connection can be canonically chosen for non-extremal WIHs. However, the method provided here provides a canonical choice regardless of whether the WIH is extremal. 

As our mapping from a NHS to a Carrollian structure fixes the $A$ parameter in the NHS structure group to unity, one must no longer make an arbitrary choice of representative in the generating null direction for the fundamental vector field. Instead, our construction chooses the unique representative with vanishing inaffinity. We note that the process leading to corollary \ref{Cor:NEH} automatically produces NEHs that are automatically extremally weakly-isolated horizons. This is always possible to do for NEHs \cite[Section III A]{Ashtekar2002}. Then, the Carrollian geometries for non-extremal horizons can be constructed using corollary \ref{cor:AchangeL} and the appropriate choice of parameter $A$ to ensure the inaffinity, acting as surface gravity, is non-zero.  In this setting, the geometry constructed here is equivalent to that of~\cite{freidel2024}.

Near-horizon geometries (NHGs) allow for the study of black holes admitting degenerate Killing horizons by generating an associated but distinct spacetime. However, this approach is ill-suited for the study of physics directly on the horizon. NHSs and the resulting Carrollian geometries provide a complementary tool to examine physics on these horizons such as Carrollian hydrodynamics and Carrollian field theories. Interestingly, the degenerate metric, ${\bf h} = 2\bcm \bar{\bcm} $ is independent of the coordinate $\rho$ and hence naturally projects to the metric on the spatial part of the horizon. Using the soliton-formalism introduced in \cite{nurowski2016generalized}, it is possible to determine all possible NHGs in vacuum and Einstein spacetimes \cite{kunduri2013classification}. In particular, this approach could give an alternative approach to classifying all associated NHGs to extremal horizons in vacuum, with a possible cosmological constant \cite{dunajski2023intrinsic}. 

The two propositions give a wealth of well motivated Carrollian geometries that can be used to study the physics of a wider range of null hypersurfaces than just degenerate Killing horizons of stationary black holes. As an immediate application, since the frame and connection are chosen uniquely, we can compute the Brown-York tensor for the corresponding torsion-free connection on null hypersurfaces \cite{Chandrasekaran2022}. This tensor provides a way for defining quasi-local gravitational charges in sub-regions bounded by null hypersurfaces. In the case of asymptotically flat spacetimes, the Brown-York charges can be computed and compared with the charges obtained for symmetries corresponding to the Bondi-Mezner-Sachs algebra, which acts as the asymptotic symmetry group at null infinity \cite{Flanagan:2015pxa, Grant:2021sxk}. This would provide further geometric information to determine explicit Carrollian geometries that describe asymptotic null infinity in vacuum spacetimes. Knowledge of these Carrollian   geometries and the potential reconstruction of the ambient vacuum spacetimes \cite{Ciambelli2018} would give further insight into holography.

While the construction provided in this manuscript has focused on Einstein spacetimes, the results can potentially be generalized to study null hypersurfaces in far more general settings. We note that if instead one considers a spacetime with a cosmological constant which also admits a null electromagnetic field whose null direction is proportional to $\bck$ the resulting field equations in \eqref{NullCartan1} and \eqref{NullCartan2Einstein} would still be applicable. This implies that the shear-free NHSs in this paper would also describe NHSs in these slightly more general spacetimes. As a more ambitious step, one may consider NHSs in asymptotically flat spacetimes with non-null EM fields as a source, such as the outer event horizon in Kerr-Newman spacetimes. Given that \cite{gray2023carrollian} employs a connection that lacks extrinsic data, more could be said about such null hypersurfaces if the appropriate frame and connection were determined.


 \begin{acknowledgments}
The authors would like to thank David Robinson who made considerable stylistic improvements to the initial draft of this paper during their collaboration with PN in the late nineties. SB is supported by the Operational Programme Research Development and Education Project No. CZ.02.01.01/00/ 22-010/0007541. DM and PN were supported by the Norwegian Financial Mechanism 2014-2021 (project registration number 2019/34/H/ST1/00636).
 \end{acknowledgments}

\bibliographystyle{unsrt}
\bibliography{NHCSbib}


\end{document}